\documentclass{amsart}
\usepackage{amssymb}
\usepackage{enumerate}
\usepackage{mathtools}
\copyrightinfo{2015}{H. De Bie \emph{et al.}}

\newtheorem{proposition}{Proposition}
\newtheorem{lemma}{Lemma}
\newtheorem{conjecture}{Conjecture}

\theoremstyle{definition}
\newtheorem{definition}{Definition}

\theoremstyle{remark}
\newtheorem{remark}{Remark}

\begin{document}
\title[The Dirac--Dunkl operator and the higher rank Bannai--Ito algebra]{The  $\mathbb{Z}_2^n$ Dirac--Dunkl operator  \\and a higher rank Bannai--Ito algebra}
\author[H. De Bie]{Hendrik De Bie}
\address{Department of Mathematical Analysis, Faculty of Engineering and Architecture, Ghent University, Galglaan 2, 9000 Ghent, Belgium}
\email{Hendrik.DeBie@UGent.be}
\author[V.X. Genest]{Vincent X. Genest}
\address{Department of Mathematics, Massachusetts Institute of Technology, 77 Massachusetts Ave, Cambridge, MA 02139, USA}
\email{vxgenest@mit.edu}
\author[L. Vinet]{Luc Vinet}
\address{Centre de Recherches Math\'ematiques, Universit\'e de Montr\'eal, P.O. Box 6128, Centre-ville Station, Montr\'eal, QC H3C 3J7, Canada}
\email{vinetl@crm.umontreal.ca}
\date{}
\dedicatory{}
\begin{abstract}
The kernel of the $\mathbb{Z}_2^{n}$ Dirac--Dunkl operator is examined. The symmetry algebra $\mathcal{A}_{n}$ of the associated Dirac--Dunkl equation on $\mathbb{S}^{n-1}$ is determined and is seen to correspond to a higher rank generalization of the Bannai--Ito algebra. A basis for the polynomial null-solutions of the Dirac--Dunkl operator is constructed. The basis elements are joint eigenfunctions of a maximal commutative subalgebra of $\mathcal{A}_{n}$ and are given explicitly in terms of Jacobi polynomials. The symmetry algebra is shown to act irreducibly on this basis via raising/lowering operators. A scalar realization of $\mathcal{A}_{n}$ is proposed.
\end{abstract}
\maketitle
\section{Introduction}
The goal of the present paper is to examine the kernel of the $n$-dimensional Dirac--Dunkl operator with $\mathbb{Z}_2^{n}$ reflection group and to study the corresponding Dirac--Dunkl equation on the $(n-1)$-sphere. The invariance algebra $\mathcal{A}_{n}$ of this equation will be obtained and identified as a higher rank generalization of the Bannai--Ito algebra  (see Definition 1). Basis functions for the space of homogeneous polynomial null solutions of the Dirac--Dunkl operator will be constructed using a Cauchy--Kovalevskaia extension theorem and will be given explicitly in terms of Jacobi polynomials. This set of basis functions, characterized as the joint eigenvectors of a maximal commutative subalgebra of $\mathcal{A}_{n}$, will be shown to transform irreducibly under $\mathcal{A}_{n}$ through the action of raising and lowering operators.

\subsection{The $\mathbb{Z}_2^{n}$ Laplace--Dunkl operator on $\mathbb{R}^{n}$} Dunkl operators are families of differential-difference operators associated with finite reflection groups. Introduced in \cite{1989_Dunkl_TransAmerMathSoc_311_167}, these operators appear in many areas: they are at the core of the study of multivariate special functions associated with root systems \cite{2001_Dunkl&Xu, 2003_MacDonald}, they arise in harmonic analysis and integral transforms \cite{2012_DeBie&Orsted&Somberg&Soucek_TransAmerMathSoc_364_3875, 2003_Rosler}, they are closely related to representations of double affine Hecke algebras \cite{2005_Cherednik}, they are at the origin of Dunkl processes \cite{2008_Graczyk&Rosler&Yor}, and they play a central role in the analysis of Calogero--Sutherland systems \cite{2000_VanDiejen&Vinet}. Consider the Abelian reflection group $\mathbb{Z}_2^{n}=\mathbb{Z}_2\times \cdots \times \mathbb{Z}_2$; the corresponding Dunkl operators $T_1, \ldots, T_{n}$ acting on $\mathbb{R}^{n}$ are defined as follows:
\begin{align*}
T_i=\partial_{x_i}+\frac{\mu_i}{x_i}(1-r_i),\qquad i=1,\ldots,n,
\end{align*}
where $\mu_1, \ldots,\mu_n$ with $\mu_i> 0$ are real parameters, $\partial_{x_i}$ is the partial derivative with respect to $x_i$ and where $r_i$ is the reflection operator in the $x_i=0$ hyperplane, i.e. $r_i f(x_i)=f(-x_i)$. It is obvious that one has $T_i T_j=T_j T_i$ for all $i,j\in \{1,\ldots, n\}$. The Laplace--Dunkl operator $\Delta$ associated to $\mathbb{Z}_2^{n}$ is defined by
\begin{align*}
\Delta=\sum_{i=1}^{n}T_i^{2}.
\end{align*}
\subsection{The $\mathbb{Z}_2^{n}$ Dirac--Dunkl operator on $\mathbb{R}^{n}$}
The Clifford algebra $\mathcal{C}\ell_n$ of negative signature is generated by the elements $e_1,\ldots, e_n$ which satisfy the defining relations
\begin{align}
\label{Clifford}
\{e_i, e_{j}\}=-2 \delta_{ij},\qquad i,j\in \{1,\ldots, n\},
\end{align}
where $\{x,y\}=xy+yx$ stands for the  anticommutator and where $\delta_{ij}$ is the Kronecker delta. The Dirac--Dunkl operator $\underline{D}$ on $\mathbb{R}^{n}$ associated to the reflection group $\mathbb{Z}_2^{n}$ and the position operator $\underline{x}$ are respectively defined as
\begin{align}
\label{Dunkl-Dirac}
\underline{D}=\sum_{i=1}^{n}e_i T_i,\qquad \underline{x}=\sum_{i=1}^{n}e_i x_i.
\end{align}
The Dirac--Dunkl operator $\underline{D}$ squares to the Laplace--Dunkl operator; indeed, it is verified that $\underline{D}^{2}=-\Delta$. One has also $\underline{x}^{2}=-\rVert x\rVert^{2}$, where $\rVert x \rVert^2=\sum_{i=1}^{n}x_i^{2}$. The ``intermediate'' Dirac operators and position operators are defined as follows. Let $A\subset [n]$, where $[n]=\{1,\ldots ,n\}$, and define
\begin{align}
\label{Dunkl-Dirac-A}
\underline{D}_{A}=\sum_{i\in A}e_i T_i,\qquad \underline{x}_{A}=\sum_{i\in A} e_i x_i.
\end{align}
In similar fashion, define
\begin{align*}
\Delta_{A}=\sum_{i\in A}T_i^{2},\qquad \rVert x_{A}\rVert^{2}=\sum_{i\in A}x_i^{2}.
\end{align*}
In this notation, the ``full'' $n$-dimensional Dirac--Dunkl operator  $\underline{D}$ given in \eqref{Dunkl-Dirac} can be written as $\underline{D}_{[n]}$; similarly, one has $\underline{x}=\underline{x}_{[n]}$. For $\ell\leq n$, we shall also use the notation $\underline{D}_{[\ell]}$ and $\underline{x}_{[\ell]}$ for $\underline{D}_{\{1,\ldots, \ell\}}$ and $\underline{x}_{\{1,\ldots,\ell\}}$, respectively. 
\subsection{The $\mathbb{Z}_2^{n}$ Dirac--Dunkl operator on $\mathbb{S}^{n-1}$}
The Dirac--Dunkl operator on the $(n-1)$-sphere and its ``intermediate'' analogs are most naturally defined in the context of the $\mathfrak{osp}(1|2)$ realizations generated by the intermediate Dirac--Dunkl operators and position operators $\underline{D}_{A}$ and $\underline{x}_{A}$. These realizations, which occur for any root system, are due to \cite{2009_Orsted&Somberg&Soucek_AdvApplCliffAlg_19_403} (see also \cite{2012_DeBie&Orsted&Somberg&Soucek_TransAmerMathSoc_364_3875}). For the particular case of $\mathbb{Z}_2^{n}$, one has the following.
\begin{proposition}[\cite{2012_DeBie&Orsted&Somberg&Soucek_TransAmerMathSoc_364_3875}]
\label{prop1}
For $A\subset [n]$, the operators $\underline{D}_A$ and $\underline{x}_{A}$ generate the Lie superalgebra $\mathfrak{osp}(1|2)$ with defining relations
\begin{alignat}{2}
\label{OSP}
\begin{aligned}
&\{\underline{x}_{A},\underline{x}_{A}\}=-2 \rVert x_{A}\rVert^{2} &\qquad &\{\underline{D}_{A}, \underline{D}_{A}\}=-2 \Delta_{A}
\\
& \{\underline{x}_{A},\underline{D}_{A}\}=-2(\mathbb{E}_{A}+\gamma_{A}) & \qquad & [\underline{D}_{A}, \mathbb{E}_{A}+\gamma_{A}]=\underline{D}_{A}
\\
&[\underline{D}_{A},\rVert x_{A}\rVert^{2}]=2 \underline{x}_{A}&\qquad&  [\mathbb{E}_{A}+\gamma_{A}, \underline{x}_{A}]=\underline{x}_{A}
\\
& [\Delta_{A}, \underline{x}_{A}]=2 \underline{D}_{A} & \qquad  &[\Delta_{A},\mathbb{E}_{A}+\gamma_A]=2 \Delta_{A}
\\
& [\Delta_{A}, \rVert x_{A}\rVert^{2}]=4(\mathbb{E}_{A}+\gamma_{A})  &\qquad & [\mathbb{E}_{A}+\gamma_{A}, \rVert x_{A}\rVert^{2}]=2 \rVert x_{A}\rVert^{2}
\end{aligned}
\end{alignat}
where $\mathbb{E}_{A}=\sum_{i\in A}x_i\partial_{x_i}$ is the Euler operator for the set $A$, $[x,y]=xy-yx$ stands for the commutator and where 
\begin{align}
\label{gamma-A}
\gamma_A=\frac{\rvert A\rvert}{2}+\sum_{i\in A} \mu_i.
\end{align}
\end{proposition}
The sCasimir operator $\underline{S}_{A}$ for the $\mathfrak{osp}(1|2)$ realization \eqref{OSP} is defined as \cite{1995_Lesniewski_JMathPhys_36_1457}
\begin{align}
\label{sCasimir}
\underline{S}_{A}=\frac{1}{2}\Big([\underline{x}_{A},\underline{D}_{A}]-1\Big).
\end{align}
As follows from \eqref{OSP}, this operator satisfies the anticommutation relations
\begin{align}
\label{S-Anti}
\{\underline{S}_{A}, \underline{D}_{A}\}=0,\qquad \{\underline{S}_{A},\underline{x}_{A}\}=0.
\end{align}
The Dirac--Dunkl operator on $\mathbb{S}^{n-1}$, which will also be referred to as the spherical Dirac--Dunkl operator, will be denoted by $\Gamma_{[n]}$. It is defined as
\begin{align}
\label{Dirac-Dunkl-S}
\Gamma_{[n]}=\underline{S}_{[n]} \prod_{i=1}^{n} r_i.
\end{align}
Since one has $\{\prod_{i\in A}r_i, \underline{x}_{A}\}=\{\prod_{i\in A}r_i, \underline{D}_{A}\}=0$, it follows from \eqref{S-Anti} that $\Gamma_{[n]}$ commutes with the full Dirac--Dunkl operator $\underline{D}$ and position operator $\underline{x}$; that is
\begin{align}
\label{Commu}
[\Gamma_{[n]}, \underline{D}]=0,\quad [\Gamma_{[n]},\underline{x}]=0.
\end{align}
Since $\Gamma_{[n]}$ also commutes with the Euler operator $\mathbb{E}_{[n]}$ on $\mathbb{R}^{n}$, it has a well defined action on the $(n-1)$-dimensional sphere. Similarly, the intermediate spherical Dirac--Dunkl operators $\Gamma_{A}$ for $A\subset [n]$ are defined by
\begin{align}
\label{Gamma-A}
\Gamma_{A}=\underline{S}_{A}\prod_{i\in A}r_{i}.
\end{align}
\subsection{Dunkl monogenics and the spherical Dirac--Dunkl equation} Consider $\mathcal{P}(\mathbb{R}^{n})=\mathbb{R}[x_1,\ldots ,x_{n}]$, the ring of polynomials in the variables $x_1,\ldots, x_{n}$. One has $\mathcal{P}(\mathbb{R}^{n})=\bigoplus_{k=0}^{\infty}\mathcal{P}_k(\mathbb{R}^{n})$, where $\mathcal{P}_k(\mathbb{R}^{n})$ is the space of homogeneous polynomials of degree $k$ in $x_1,\ldots, x_{n}$. Let $V$ be a representation space for $\mathcal{C}\ell_{n}$, fixed once and for all. For example, one could fix $V$ as the spinor space or as the Clifford algebra itself. The space $\mathcal{M}_k(\mathbb{R}^{n};V)$ of Dunkl monogenics of degree $k$ associated with the reflection group $\mathbb{Z}_2^{n}$ is defined as
\begin{align*}
\mathcal{M}_{k}(\mathbb{R}^{n};V)=\mathrm{Ker}\;\underline{D}\,\cap\,(\mathcal{P}_k(\mathbb{R}^{n})\otimes V).
\end{align*}
Similarly, the space $\mathcal{H}_k(\mathbb{R}^{n})$ of Dunkl harmonics of degree $k$ is defined as $\mathcal{H}_k(\mathbb{R}^{n})=\mathrm{Ker}\,\Delta\, \cap\, \mathcal{P}_k(\mathbb{R}^{n})$. The space of $V$-valued Dunkl harmonics of degree $k$ decomposes as follows \cite{2009_Orsted&Somberg&Soucek_AdvApplCliffAlg_19_403}:
\begin{align*}
\mathcal{H}_k(\mathbb{R}^{n})\otimes V=\mathcal{M}_k(\mathbb{R}^{n};V)\oplus \underline{x}\,\mathcal{M}_{k-1}(\mathbb{R}^{n};V).
\end{align*}
This leads to the Fischer decomposition of $V$-valued homogeneous polynomials.
\begin{proposition}[\cite{2009_Orsted&Somberg&Soucek_AdvApplCliffAlg_19_403}]
The space $\mathcal{P}_k(\mathbb{R}^{n})\otimes V$ of $V$-valued homogeneous polynomials of degree $k$ has the direct sum decomposition
\begin{align*}
\mathcal{P}_k(\mathbb{R}^{n})\otimes V=\bigoplus_{j=0}^{k}\underline{x}^{j}\,\mathcal{M}_{k-j}(\mathbb{R}^{n};V).
\end{align*}
\end{proposition}
The next proposition states that the space of Dunkl monogenics $\mathcal{M}_k(\mathbb{R}^{n};V)$ is an eigenspace for the spherical Dirac--Dunkl operator $\Gamma_{[n]}$. This means  in particular that Dunkl monogenics of degree $k$ satisfy a Dirac--Dunkl equation on $\mathbb{S}^{n-1}$.
\begin{proposition}
Let $\Psi_k\in \mathcal{M}_k(\mathbb{R}^{n};V)$ be a monogenic polynomial of degree $k$, then $\Psi_{k}$ satisfies the spherical Dirac--Dunkl equation
\begin{align}
\label{Dunkl-Dirac-Eq-S}
\Gamma_{[n]}\,\Psi_{k}=(-1)^{k}(k+\gamma_{[n]}-1/2)\Psi_{k},
\end{align}
where $\gamma_{[n]}$ is as in \eqref{gamma-A}.
\end{proposition}
\begin{proof}
By a direct calculation, one finds
\begin{align*}
\Gamma_{[n]}\Psi_{k}&=\frac{(-1)^{k}}{2}\big(\underline{x}_{[n]}\underline{D}_{[n]}-\underline{D}_{[n]}\underline{x}_{[n]}-1\big)\Psi_{k}=\frac{(-1)^{k}}{2}\big(-\underline{D}_{[n]}\underline{x}_{[n]}-1\big)\Psi_{k}
\\
&=\frac{(-1)^{k}}{2}\big(-\underline{x}_{[n]}\underline{D}_{[n]}-\underline{D}_{[n]}\underline{x}_{[n]}-1\big)\Psi_{k}
=\frac{(-1)^{k+1}}{2}\big(\{\underline{x}_{[n]},\underline{D}_{[n]}\}+1\big)\Psi_{k}
\\
&=\frac{(-1)^{k+1}}{2}\big(-2(\mathbb{E}_{[n]}+\gamma_{[n]})+1\big)\Psi_{k}=(-1)^{k}(k+\gamma_{[n]}-1/2)\Psi_{k},
\end{align*}
where we successively used the definition \eqref{Dirac-Dunkl-S}, the relation $(\prod_{j=1}^{n}r_j) p=(-1)^{k}p$ valid for $p\in \mathcal{P}_k(\mathbb{R}^{n})$, the kernel property $\underline{D}\Psi_k=0$, the $\mathfrak{osp}(1|2)$ relations \eqref{OSP} and the relation $\mathbb{E}_{[n]}q=k\,q$ valid for $q\in \mathcal{P}_k(\mathbb{R}^{n})$.
\end{proof}
One of the main results of this paper is the presentation and characterization of the symmetry algebra of the Dirac--Dunkl equation \eqref{Dunkl-Dirac-Eq-S} associated with the kernel of the $n$-dimensional Dirac--Dunkl operator $\underline{D}$. This novel abstract structure, defined in Section 2, will be identified as a higher rank generalization of the Bannai--Ito algebra. Let us now introduce the basics of this algebra in the rank-one case.

\subsection{The Bannai--Ito algebra}
The Bannai--Ito algebra was originally introduced in \cite{2012_Tsujimoto&Vinet&Zhedanov_AdvMath_229_2123} as the algebraic structure underlying the bispectral property of the Bannai--Ito polynomials, which sit atop of the hierarchy of ``$-1$'' orthogonal polynomials \cite{2013_Genest&Vinet&Zhedanov_SIGMA_9_18, 2012_Tsujimoto&Vinet&Zhedanov_AdvMath_229_2123}. This algebra has generators $A_1, A_2, A_3$ and defining relations
\begin{align}
\label{BI}
\{A_1,A_2\}=A_3+\alpha_3,\qquad \{A_2,A_3\}=A_1+\alpha_1,\qquad \{A_3,A_1\}=A_2+\alpha_2,
\end{align}
where $\alpha_1, \alpha_2,\alpha_3$ are central elements. The algebra \eqref{BI} has appeared in \cite{2014_Genest&Vinet&Zhedanov_ProcAmMathSoc_142_1545} as a duality algebra for the Racah problem of $\mathfrak{osp}(1|2)$ (then called $sl_{-1}(2)$) and in \cite{2014_Genest&Vinet&Zhedanov_JPhysA_47_205202, 2015_Genest&Vinet&Zhedanov_CommMathPhys_336_243} as an invariance algebra for superintegrable systems with reflections; when $\alpha_i=0$, it can be viewed as the $q\rightarrow -1$ limit of $U_q'(\mathfrak{so}(3))$ \cite{1991_Gavrilik&Klimyk_LettMathPhys_21_215}. Recently, it has also been recognized that the Bannai--Ito algebra corresponds to a degeneration of the rank-one double affine Hecke algebra of type $(\check{C}_1, C_1)$ \cite{Genest-TMP-2015}.

In the paper \cite{2015_DeBie&Genest&Vinet_ArXiv_1501.03108} by the authors, it was shown that the Bannai--Ito algebra is the symmetry algebra of the Dirac--Dunkl equation on the 2-sphere associated to the $\mathbb{Z}_2^{3}$ reflection group. In that context, relevant finite-dimensional representations of \eqref{BI} were constructed, explicit formulas for the basis vectors of these representations were found, and a connection with the Bannai--Ito polynomials was established.

\subsection{Goal and outline} The main goal of the present paper is to extend the results of \cite{2015_DeBie&Genest&Vinet_ArXiv_1501.03108} to arbitrary dimension. As shall be seen this extension is involved, and yet the results are instructive and lend themselves to an elegant presentation. The outline is as follows. In Section two, the symmetry algebra of the Dirac--Dunkl equation is investigated. In Section three, an explicit basis for the space of Dunkl monogenics is constructed. In Section four, the higher rank Bannai--Ito algebra obtained in Section two is shown to act irreducibly on the basis found in Section three. In Section five, a scalar version of the model is presented. A short conclusion follows.

\section{Symmetry algebra of the Dirac--Dunkl operator}
In this section, the joint symmetries of the Dirac--Dunkl and spherical Dirac--Dunkl operators are obtained. These symmetries are shown to close under anticommutation. The resulting invariance algebra $\mathcal{A}_{n}$ is interpreted as a higher rank extension of the Bannai--Ito algebra. Two Casimir operators are exhibited.

Let us begin by showing that the intermediate spherical Dirac--Dunkl operators $\Gamma_{A}$ for $A\subset [n]$ introduced in \eqref{Gamma-A} are in fact symmetries of $\underline{D}$ and $\Gamma_{[n]}$.

\begin{lemma}
For $A\subset [n]$, the operator $\Gamma_{A}$ defined in \eqref{Gamma-A} satisfies
\begin{enumerate}[i)]
\item $[\Gamma_{A},\underline{D}]=[\Gamma_{A},\underline{x}]=0$
\item $[\Gamma_{A}, \Gamma_{[n]}]=0$
\end{enumerate}
where $\underline{D}$ and $\underline{x}$ are  as in \eqref{Dunkl-Dirac}.
\end{lemma}
\begin{proof}
For $A\subset[n]$, one has $\underline{D}=\underline{D}_{A}+\underline{D}_{[n]\setminus A}$ and $\underline{x}=\underline{x}_{A}+\underline{x}_{[n]\setminus A}$. Consequently, one can write
\begin{align*}
[\Gamma_{A},\underline{D}]=[\Gamma_{A},\underline{D}_{A}]+[\Gamma_{A},\underline{D}_{[n]\setminus A}]=[\Gamma_{A},\underline{D}_{A}]=0,
\end{align*}
where the relations \eqref{Clifford} and the property \eqref{S-Anti} were used; the result for $\underline{x}$ is proven in a similar fashion, yielding $i)$. In view of the definition \eqref{Gamma-A} of $\Gamma_{A}$, $ii)$ is seen to follow directly from $i)$.
\end{proof}
\begin{remark}
When $A=\emptyset$ or $A=\{k\}$, it follows from the definition \eqref{Gamma-A} that
\begin{align}
\label{Obs}
\Gamma_{\emptyset}=-\frac{1}{2},\qquad \Gamma_{\{k\}}=\mu_k.
\end{align}
\end{remark}
For two distinct subsets $A,B\subset [n]$, it is seen that the operators $\Gamma_{A}$ and $\Gamma_{B}$ will generally not commute. The symmetries $\Gamma_{A}$ for $A\subset[n]$ hence generate the non-Abelian invariance (or symmetry) algebra of the Dirac--Dunkl and spherical Dirac--Dunkl operators. For concreteness, we give below the explicit expression for the intermediate spherical Dirac--Dunkl operators. For $A\subset{[n]}$, $\Gamma_{A}$ has the form
\begin{align}
\label{Gamma-A-Explicit}
\Gamma_{A}=\Big(\smashoperator{\sum_{\{i,j\}\subset A}} M_{ij}+\frac{|A|-1}{2}+\sum_{k\in A}\mu_k r_k\Big) \prod_{i\in A} r_i,
\end{align}
where the first summation is performed on all 2-subsets of $A$ and where $M_{ij}$ reads
\begin{align}
\label{M}
M_{ij}=e_i e_j (x_{i}T_{j}-x_{j} T_{i}).
\end{align}
As is seen from \eqref{Gamma-A-Explicit}, $\Gamma_{A}$ is invariant under permutations of the indices. 

We now prove a series of lemmas unveiling the commutation relations between the symmetry operators $\Gamma_{A}$.
\begin{lemma}
For $A\subset [n]$, $j\in A$ and $k\notin A$, one has
\begin{align}
\label{Comm-1}
\{\Gamma_{A}, \Gamma_{\{j,k\}}\}=\Gamma_{(A\cup \{k\})\setminus \{j\}}+2\,\Gamma_{\{j\}}\Gamma_{A\cup \{k\}}+2\,\Gamma_{\{k\}} \Gamma_{A\setminus\{j\}}.
\end{align}
\end{lemma}
\begin{proof}
The proof of \eqref{Comm-1} follows by a direct calculation using \eqref{Gamma-A-Explicit}. For simplicity, one can perform the calculation with $A=[\ell]$ and $\{j,k\}=\{1,\ell+1\}$ and extend the result by symmetry. Proceeding with this calculation, it is useful to observe that
\begin{align*}
\{M_{ij}r_i r_j, M_{jk}r_j r_k\}=(1+2 \mu_{j} r_j) M_{ik}r_{i} r_{k},\qquad i\neq k.
\end{align*} 
\end{proof}
\begin{lemma}
\label{Lemma-Comm}
For $A\subset B$, $B\subset A$ or $A\cap B=\emptyset$, one has
\begin{align*}
[\Gamma_{A}, \Gamma_{B}]=0.
\end{align*}
\end{lemma}
\begin{proof}
When $A\subset B$ or $B\subset A$, the result follows directly from lemma 1, with $\underline{D}$, $\underline{x}$ and $\Gamma_{[n]}$ replaced by $\underline{D}_{B}$, $\underline{x}_{B}$ and $\Gamma_{B}$ or $\underline{D}_{A}$, $\underline{x}_{A}$ and $\Gamma_{A}$,  respectively. When $A\cap B=\emptyset$, the result is directly obtained using the explicit expression \eqref{Gamma-A-Explicit} for the intermediate spherical Dirac--Dunkl operators.
\end{proof}

\begin{lemma}
For $A, B\subset [n]$ such that $A\cap B=\{k\}$, one has
\begin{align}
\label{Comm-2}
\{\Gamma_{A}, \Gamma_{B}\}=\Gamma_{(A\cup B)\setminus \{k\}}+2\,\Gamma_{\{k\}}\Gamma_{A\cup B}+2\,\Gamma_{A\setminus\{k\}}\Gamma_{B\setminus\{k\}}.
\end{align}
\end{lemma}
\begin{proof}
By induction on $|B|$. When $|B|=1$, the result follows from \eqref{Obs}. When $|B|=2$, the result holds by Lemma 2. Suppose that the result holds for $|B|=n$ and consider the set $\widetilde{B}=B\cup \{x\}$ with $x\notin A$, $A\cap B=\{k\}$ and $|B|=n$. Let $y\in B$ with $y\notin A$, it follows from Lemma 2 that
\begin{align}
\label{inter1}
\Gamma_{\widetilde{B}}=\frac{1}{2\mu_y}\left(\{\Gamma_{B}, \Gamma_{\{x,y\}}\}-\Gamma_{\widetilde{B}\setminus \{y\}}-2\mu_x \Gamma_{B\setminus \{y\}}\right).
\end{align}
Using the fact that $[\Gamma_{A}, \Gamma_{\{x,y\}}]=0$ and the identity $\{A,\{B,C\}\}=\{\{A,B\},C\}$ which holds when $[A,C]=0$, one can write
\begin{align}
\label{inter2}
\{\Gamma_{A}, \Gamma_{\widetilde{B}}\}=\frac{1}{2\mu_y}\Big(\{\{\Gamma_{A}, \Gamma_{B}\}, \Gamma_{\{x,y\}}\}-\{\Gamma_{A},\Gamma_{\widetilde{B}\setminus \{y\}}\}-2\mu_x \{\Gamma_{A}, \Gamma_{B\setminus \{y\}}\}\Big).
\end{align}
Each anticommutator appearing in the above expression can be evaluated using the induction hypothesis. Upon expanding the result, one directly finds \eqref{Comm-2} with $B$ replaced by $\widetilde{B}$. This completes the induction.
\end{proof}
\begin{lemma}
For $A,B\subset [n]$ such that $A\cap B\neq \emptyset$, one has
\begin{align}
\label{Comm-3}
\{ \Gamma_{A}, \Gamma_{B}\}=\Gamma_{(A\cup B)\setminus (A\cap B)}+2\,\Gamma_{A\cap B}\Gamma_{A\cup B}+2\, \Gamma_{A\setminus( A\cap B)}\Gamma_{B\setminus(A\cap B)}.
\end{align}
\end{lemma}
\begin{proof}
By induction on $|A\cap B|$. When $|A\cap B|=1$, the result holds by virtue of Lemma 4. Suppose that \eqref{Comm-3} holds at level  $m$ and consider the set $\widetilde{B}=B \cup \{x\}$ with $x\in A$ and $|A\cap B|=m$. Consider $y\in A\cap B$; one can write $\Gamma_{\widetilde{B}}$ as in \eqref{inter1}. Since $x,y\in A$, $\Gamma_{A}$ commutes with $\Gamma_{\{x,y\}}$ and one can write $\{\Gamma_{A},\Gamma_{\widetilde{B}}\}$ as in \eqref{inter2}, the only difference with \eqref{inter2} being that $x,y\in A$. Upon applying the induction hypothesis, one finds \eqref{Comm-3} with $B$ replaced by $\widetilde{B}$ after a straightforward calculation.
\end{proof}
The results of the five preceding Lemmas can now be combined to give the complete set of relations between the symmetries $\Gamma_{A}$.
\begin{proposition}
The symmetries $\Gamma_{A}$ with $A\subset [n]$ of the Dirac--Dunkl and spherical Dirac--Dunkl operators satisfy the anticommutation relations
\begin{align}
\label{BI-Relations}
\{\Gamma_{A},\Gamma_{B}\}=\Gamma_{(A\cup B)\setminus (A\cap B)}+2\,\Gamma_{A\cap B}\Gamma_{A\cup B}+2\,\Gamma_{A\setminus (A\cap B)}\Gamma_{B\setminus(A\cap B)}.
\end{align}
\begin{proof}
The proposition follows from the combination of Lemmas 2--5.
\end{proof}
\end{proposition}
This result motivates the following definition for the abstract symmetry algebra of the spherical Dirac--Dunkl operator $\Gamma_{[n]}$, which shall be denoted by  $\mathcal{A}_{n}$.
\begin{definition}
Let $n\geq 3$ and  $A\subset [n]$. Call $\mathcal{A}_{n}$ the abstract associative algebra with generators $\Gamma_{A}$ and defining relations \eqref{BI-Relations}.
\end{definition}
\begin{remark}
It is visible from Definition 1 that one has $\mathcal{A}_{n}\subset\mathcal{A}_{n+1}$. That is $\mathcal{A}_{n}$ is a subalgebra of $\mathcal{A}_{n+1}$.
\end{remark}

We now show that operators of the form $\Gamma_{\{i,j\}}$ are sufficient to generate the entire invariance algebra .
\begin{lemma}
\label{Gen-Set}
The set of operators $\Gamma_{B}$ for which $B$ is a 2-subset of $[n]$ constitutes a generating set for  $\mathcal{A}_{n}$.
\end{lemma}
\begin{proof}
First observe that any generator $\Gamma_{A}$ with $|A|=m$  where $m>2$ can be written as a polynomial in generators $\Gamma_{A'}, \Gamma_{A''}$  with $|A'|=m-1$ and $|A''|=2$. This can be done (non uniquely) by choosing two sets $C$ and $D$ with $|C|=m-1$ and $|D|=2$ such that $C\cup D=A$ and $C\cap D=\{k\}$, and using the relation
\begin{align*}
\Gamma_{A}=\Gamma_{C\cup D}=\frac{1}{2\mu_k}\left(\{\Gamma_{C},\Gamma_{D}\}-\Gamma_{(C\cup D)\setminus (C\cap D)}-2\,\Gamma_{C\setminus (C\cap D)}\Gamma_{D\setminus(C\cap D)}\right).
\end{align*}
Applying this procedure recursively, one can write any generator $\Gamma_{A}$ with $|A|>2$ as a polynomial in generators $\Gamma_{B}$ with $|B|=2$.
 \end{proof}
\begin{remark}
The algebra $\mathcal{A}_{n}$ can be considered as a higher rank generalization of the Bannai--Ito algebra \eqref{BI}. To see this, consider the case $n=3$. The symmetry algebra $\mathcal{A}_3$ of the spherical Dirac--Dunkl operator in three dimensions $\Gamma_{[3]}$ is generated by the operators $K_3=\Gamma_{\{1,2\}}$,  $K_1=\Gamma_{\{2,3\}}$ and  $K_2=\Gamma_{\{1,3\}}$ with relations
\begin{align*}
\{K_1,K_2\}=K_3+\omega_3,\qquad \{K_2,K_3\}=K_1+\omega_1,\qquad \{K_3,K_1\}=K_2+\omega_2,
\end{align*}
where $\omega_1$, $\omega_2$, $\omega_3$ are given by
\begin{align*}
\omega_1=2 \mu_1 \Gamma_{[3]}+2\mu_2\mu_3,\quad \omega_2=2\mu_2\Gamma_{[3]}+2\mu_1\mu_3,\quad \omega_3=2\mu_3 \Gamma_{[3]}+2\mu_1\mu_2.
\end{align*}
It is clear that $\omega_1$, $\omega_2$ and $\omega_3$ are central elements. Moreover, if one considers the realization of the algebra on the space of Dunkl monogenics $\mathcal{M}_{k}(\mathbb{R}^3; V)$, one can use the eigenvalue equation \eqref{Dunkl-Dirac-Eq-S} to obtain $\mathbb{R}$-valued $\omega_i$'s. For a detailed analysis of the $n=3$ case, the reader can consult \cite{2015_DeBie&Genest&Vinet_ArXiv_1501.03108}.
\end{remark}
The algebra $\mathcal{A}_{n}$ has an important Abelian subalgebra $\mathcal{Y}_{n}$. This subalgebra is generated by the $(n-2)$ pairwise commuting elements
\begin{align}
\label{Abelian-Sub}
\mathcal{Y}_{n}=\langle \Gamma_{[2]}, \Gamma_{[3]},\ldots, \Gamma_{[n-1]}\rangle.
\end{align}
Given that $[n-1]\subset [n]$, the commutativity of $\mathcal{Y}_{n}$ follows directly from Lemma \ref{Lemma-Comm}. Note that $\Gamma_{\emptyset}$, $\Gamma_{[1]}$ and $\Gamma_{[n]}$ are not included in the generating set of $\mathcal{Y}_{n}$, as they are central in $\mathcal{A}_{n}$. It is manifest from the defining relations \eqref{BI-Relations} that $\mathcal{Y}_{n}$ is the largest possible Abelian subalgebra of $\mathcal{A}_{n}$. Thus, borrowing from the terminology of Lie algebras, one can say that the algebra $\mathcal{A}_{n}$ has rank $(n-2)$. Observe that one can easily define another maximal Abelian subalgebra by applying a permutation of $S_{n}$ on $\mathcal{Y}_n$; this is done by taking $\pi \mathcal{Y}_{n}=\langle \Gamma_{\pi [2]}, \Gamma_{\pi[3]},\ldots, \Gamma_{\pi[n-1]}\rangle $ for $\pi \in S_{n}$. For example, applying the permutation $\pi=(123\cdots n)$ on $\mathcal{Y}_n$ gives an algebra that we shall denote by $\mathcal{Z}_{n}$
\begin{align}
\label{Abelian-Sub-2}
\mathcal{Z}_{n}=\langle \Gamma_{\{2,3\}}, \Gamma_{\{2,3,4\}},\ldots, \Gamma_{\{2,3,\ldots,n\}}\rangle.
\end{align}
\begin{remark}
While $\mathcal{A}_{3}$ coincides with a degenerate double affine Hecke algebra of type $(\check{C}_1, C_1)$ \cite{Genest-TMP-2015}, it appears that $\mathcal{A}_{n}$ does not coincide with the higher rank version of this degenerate double affine Hecke algebra of type $(\check{C}_{n-2}, C_{n-2})$, which has been investigated in \cite{2007_Groenevelt_TransGroups_12_77, 2009_Groenevelt_SelMath_15_377}.
\end{remark}
We shall now exhibit two non-trivial central elements of the algebra $\mathcal{A}_{n}$. These two elements, which shall be referred to as Casimir operators, are expressed as quadratic combinations of the generators.
\begin{lemma}
Let $Q_{[n]}$ be the element of $\mathcal{A}_{n}$ defined as
\begin{align}
\label{Cas1}
Q_{[n]}=\smashoperator{\sum_{\{i,j\}\subset [n]}}\Gamma_{\{i,j\}}^{2}.
\end{align}
Then for $A\subset [n]$ one has  $[Q_{[n]}, \Gamma_{A}]=0$.
\end{lemma}
\begin{proof}
In view of Lemma \ref{Gen-Set}, it suffices to prove that $Q_{[n]}$ commutes with every generator of the from $\Gamma_{\{i,k\}}$. Using the commutator identity $[A^2,C]=[A,\{A,C\}]$ and the relations \eqref{BI-Relations}, one finds
\begin{align*}
[Q_{[n]}, \Gamma_{\{i,k\}}]&=\smashoperator{\sum_{\substack{j\in [n]\\ j\neq k\neq i}}}\,\big([\Gamma_{\{i,j\}}^2, \Gamma_{\{i,k\}}]+[\Gamma_{\{j,k\}}^{2},\Gamma_{\{i,k\}}]\big)
\\
&=\smashoperator{\sum_{\substack{j\in [n]\\ j\neq k\neq i}}}\,\big([\Gamma_{\{i,j\}}, \Gamma_{\{j,k\}}]+[\Gamma_{\{j,k\}},\Gamma_{\{i,j\}}]\big)=0.
\end{align*}
\end{proof}
Since $\mathcal{A}_{n}\subset \mathcal{A}_{n+1}$, the previous can be equivalently formulated as follows.
\begin{lemma}
Let $Q_{A}$ be the element defined as 
\begin{align}
\label{QA}
Q_{A}=\smashoperator{\sum_{\{i,j\}\subset A}}\Gamma_{\{i,j\}}^2,
\end{align}
one then has $[Q_{A}, \Gamma_{B}]=0$, for all $B\subset A$.
\end{lemma}

In the realization \eqref{Gamma-A-Explicit}, the Casimir operator $Q_{A}$ is related to the spherical Dunkl-Dirac operator $\Gamma_{A}$. The result is as follows.
\begin{lemma}
\label{Lem-Relation}
For $A\subset [n]$, one has
\begin{align}
\label{Relation}
Q_{A}=\Gamma_{A}^{2}+(|A|-2)\sum_{i\in A}\mu_i^2-\frac{1}{8}(|A|-1)(|A|-2).
\end{align}
\end{lemma}
\begin{proof}
The result is proved by a direct calculation using the explicit expression \eqref{Gamma-A-Explicit} for the operators $\Gamma_{A}$. On the one hand, one has 
\begin{align*}
\Gamma_{\{i,j\}}^2= M_{ij}^2+M_{ij}+(\mu_i r_i+\mu_j r_j+1/2)^2,
\end{align*}
which gives
\begin{align}
\label{A1}
Q_{A}=\smashoperator{\sum_{\{i, j\}\subset A}}\Big[M_{ij}^2+M_{ij}+2\mu_i\mu_jr_i r_j\Big]
+(|A|-1)\sum_{i\in A}(\mu_i r_i+\mu_i^2)+\textstyle{\frac{|A|(|A|-1)}{8}}.
\end{align} 
On the other hand, one has
\begin{multline*}
\Gamma_{A}^2=\smashoperator{\sum_{\{i,j\}\subset A}}\Big[M_{ij}^2+\smashoperator{\sum_{\substack{k\in A\\ k\neq i\neq j}}}\Big(\frac{1}{2}[\{M_{ij}, M_{ik}\}+\{M_{ij},M_{jk}\}]+2\mu_k M_{ij}r_{k}\Big)\Big]
\\[-.2cm]
+(|A|-1) \smashoperator{\sum_{\{i,j\}\subset A}}M_{ij}+\big[\textstyle{\frac{|A|-1}{2}}\displaystyle+\sum_{k\in A}\mu_k r_k\big]^2.
\end{multline*}
Since $\{M_{ij},M_{ik}\}=-(1+2\mu_i R_i) M_{jk}$, one obtains
\begin{align}
\label{A2}
\Gamma_{A}^2=\smashoperator{\sum_{\{i,j\}\subset A}}\Big(M_{ij}^2+M_{ij}\Big)+\Big(\textstyle{\frac{|A|-1}{2}}\displaystyle+\sum_{k\in A}\mu_k r_k\Big)^2.
\end{align}
Upon expanding the last term on the right-hand side of \eqref{A2} and comparing with \eqref{A1}, one readily finds \eqref{Relation}.
\end{proof}
When $n\geq 4$, one more Casimir operator of $\mathcal{A}_{n}$ can be identified.
\begin{lemma}
Let $C_{[n]}$ be defined as
\begin{align*}
C_{[n]}=\smashoperator{\sum_{\{i,j\}\subset [n]}}\,\Gamma_{\{i,j\}}\Gamma_{[n]\setminus \{i,j\}}-(n-2)\sum_{i\in [n]}\mu_i\, \Gamma_{[n]\setminus \{i\}}.
\end{align*}
For $A\subset [n]$ one has $[C_{[n]}, \Gamma_{A}]=0$.
\end{lemma}
\begin{proof}
In view of Lemma \ref{Gen-Set}, it suffices to prove that $C_{[n]}$ commutes with any generator $\Gamma_{A}$ with $|A|=2$. One has
\begin{multline*}
[C_{[n]}, \Gamma_{\{i,k\}}]=\smashoperator{\sum_{\substack{j\in [n]\\ j\neq i\neq k}}}[\Gamma_{\{i,j\}}\Gamma_{[n\setminus\{i,j\}]}, \Gamma_{\{i,k\}}]+[\Gamma_{\{j,k\}}\Gamma_{[n\setminus\{j,k\}]}, \Gamma_{\{i,k\}}]
\\
-(n-2)\mu_i\;[\Gamma_{[n]\setminus \{i\}}, \Gamma_{\{i,k\}}]-(n-2)\mu_{k}\;[\Gamma_{[n]\setminus \{k\}}, \Gamma_{\{i,k\}}].
\end{multline*}
Using the identity $[AB,C]=A\{B,C\}-\{A,C\}B$ and \eqref{BI-Relations}, one finds
\begin{multline*}
[C_{[n]}, \Gamma_{\{i,k\}}]=-(n-2)\mu_i\;[\Gamma_{[n]\setminus \{i\}}, \Gamma_{\{i,k\}}]-(n-2)\mu_{k}\;[\Gamma_{[n]\setminus \{k\}}, \Gamma_{\{i,k\}}]
\\
2\mu_i\smashoperator{\sum_{\substack{j\in [n]\\ j\neq i\neq k}}}\Gamma_{\{i,j\}}\Gamma_{[n]\setminus\{i,j,k\}}+\Gamma_{\{j,k\}}\Gamma_{[n]\setminus\{j\}}-\Gamma_{\{i,j,k\}}\Gamma_{[n]\setminus \{i,j\}}-\mu_j \Gamma_{[n]\setminus \{j,k\}}
\\
+2 \mu_k \smashoperator{\sum_{\substack{j\in [n]\\ j\neq i\neq k}}}\Gamma_{\{j,k\}}\Gamma_{[n]\setminus\{i,j,k\}}+\Gamma_{\{i,j\}}\Gamma_{[n]\setminus\{j\}}-\Gamma_{\{i,j,k\}}\Gamma_{[n]\setminus \{j,k\}}-\mu_j \Gamma_{[n]\setminus \{i,j\}}
\end{multline*}
Recalling that
\begin{align*}
\Gamma_{[n]\setminus \{i\}}=\frac{1}{2\mu_k}\Big[\{\Gamma_{[n]\setminus \{i,j\}}, \Gamma_{\{j,k\}}\}-\Gamma_{[n]\setminus \{i,k\}}-2\mu_j \Gamma_{[n]\setminus\{i,j,k\}}\Big],
\end{align*}
one can write
\begin{align*}
(n-2)\mu_i[\Gamma_{[n]\setminus \{i\}}, \Gamma_{\{i,k\}}]=\frac{\mu_i}{2\mu_k}\smashoperator{\sum_{\substack{j\in [n]\\ j\neq i\neq k}}}[\{\Gamma_{[n]\setminus \{i,j\}}, \Gamma_{\{j,k\}}\}, \Gamma_{\{i,k\}}].
\end{align*}
With the help of the identity $[\{A,B\}, C]=[A,\{B,C\}]+[B,\{A,C\}]$ and the relations \eqref{BI-Relations}, one finds
\begin{align*}
(n-2)\mu_i[\Gamma_{[n]\setminus \{i\}}, \Gamma_{\{i,k\}}]=\mu_i\smashoperator{\sum_{\substack{j\in [n]\\ j\neq i\neq k}}}[\Gamma_{[n]\setminus \{i,j\}}, \Gamma_{\{i,j,k\}}]+[\Gamma_{\{j,k\}},\Gamma_{[n]\setminus \{j\}}].
\end{align*}
Using this result, one finds
\begin{multline*}
[C_{[n]}, \Gamma_{\{i,k\}}]=
\\\mu_i\smashoperator{\sum_{\substack{j\in [n]\\ j\neq i\neq k}}}\{\Gamma_{\{j,k\}},\Gamma_{[n]\setminus\{j\}}\}-\{\Gamma_{\{i,j,k\}},\Gamma_{[n]\setminus\{i,j\}}\}+2\Gamma_{\{i,j\}}\Gamma_{[n]\setminus\{i,j,k\}}-2\mu_j \Gamma_{[n]\setminus \{j,k\}}
\\
+\mu_k\smashoperator{\sum_{\substack{j\in [n]\\ j\neq i\neq k}}}\{\Gamma_{\{i,j\}},\Gamma_{[n]\setminus\{j\}}\}-\{\Gamma_{\{i,j,k\}},\Gamma_{[n]\setminus\{j,k\}}\}+2\Gamma_{\{j,k\}}\Gamma_{[n]\setminus\{i,j,k\}}-2\mu_j \Gamma_{[n]\setminus \{i,j\}}
\end{multline*}
Using \eqref{BI-Relations} one last time confirms that $[C_{[n]},\Gamma_{A}]=0$ when $|A|=2$, and thus that $[C_{[n]},\Gamma_{A}]$ for all $A\subset [n]$.
\end{proof}
The previous Lemma is equivalent to the following.
\begin{lemma}
Let $C_{A}$ be defined as 
\begin{align*}
C_{A}=\smashoperator{\sum_{\{i,j\}\subset A}}\,\Gamma_{\{i,j\}}\Gamma_{A\setminus \{i,j\}}-(|A|-2)\sum_{i\in A}\mu_i\, \Gamma_{A\setminus \{i\}}.
\end{align*}
For $B\subset A$ one has $[C_{A}, \Gamma_{B}]=0$.
\end{lemma}
In the realization \eqref{Gamma-A-Explicit}, the operator $C_{A}$ can also be expressed in terms of $\Gamma_{A}$.
\begin{lemma}
For $A\subset [n]$, one has
\begin{align}
\label{Cas2-Value}
C_{A}=\left(\frac{|A|(|A|-3)}{4}\right)\Gamma_{A}.
\end{align}
\end{lemma}
\begin{proof}
The proof proceeds by a direct calculation using the realization \eqref{Gamma-A-Explicit} of the operators $\Gamma_{A}$. The main observation is that for a given set $A$, one has
\begin{align*}
\sum_{\{i,j\}\subset A}\sum_{\substack{\{k,\ell\}\subset A\\ \{k,\ell\}\neq \{i,j\}}} M_{ij}M_{k\ell}=0.
\end{align*}
This result follows directly from the Clifford relations \eqref{Clifford} and the definition of the operators $M_{ij}$ given in \eqref{M}.
Using this result and combinatorial arguments yields the expression
\begin{multline*}
\smashoperator{\sum_{\{i,j\}\subset A}}\,\Gamma_{\{i,j\}}\Gamma_{A\setminus \{i,j\}}=
\\
\Big[\textstyle{\frac{|A|(|A|-3)}{4}}\displaystyle\smashoperator{\sum_{\{i,j\}\subset A}} M_{ij}+(|A|-2)\sum_{i\in A}\mu_i\;\smashoperator{\sum_{\substack{\{k,\ell\}\subset A\\ k,\ell\neq i}}}M_{k\ell}
+2(|A|-2)\smashoperator{\sum_{\{i,j\}\subset A}} \mu_i \mu_j r_i r_{j}
\\
+\textstyle{\frac{(|A|-1)(3|A|-8)}{4}}\displaystyle\sum_{j\in A}\mu_j r_j
+\textstyle{\frac{|A|(|A|-1)(|A|-3)}{8}}\displaystyle\Big] \prod_{i\in A}r_i.
\end{multline*}
Upon expanding the second term $\sum_{i\in A}\mu_i\, \Gamma_{A\setminus \{i\}}$  of the operator $C_{A}$ using \eqref{Gamma-A-Explicit}, a straightforward calculation shows that one has indeed \eqref{Cas2-Value}.

\end{proof}
\section{A basis for the space of Dunkl monogenics}
In this section, a basis for the space $\mathcal{M}_{k}(\mathbb{R}^{n};V)$ of polynomial null-solutions of the Dirac--Dunkl operator $\underline{D}_{[n]}$ is constructed and explicitly presented in terms of Jacobi polynomials. The construction proceeds from a Cauchy-Kovalevskaia (CK) extension theorem and from the Fischer decomposition. It is shown that the basis functions stemming from that construction are joint eigenfunctions for the generators of the  Abelian subalgebra $\mathcal{Y}_{n}$ and the corresponding eigenvalues are computed explicitly. Another basis for $\mathcal{M}_{k}(\mathbb{R}^{n};V)$ corresponding to the Abelian subalgebra $\mathcal{Z}_{n}$ defined in \eqref{Abelian-Sub-2} is introduced. We conjecture that the expansion coefficients between the two bases are given in terms of some multivariate Bannai--Ito polynomials.

\subsection{Basis functions and the CK isomorphism} There is an isomorphism $\mathbf{CK}_{x_j}^{\mu_j}: \mathcal{P}_{k}(\mathbb{R}^{j-1})\otimes V\rightarrow \mathcal{M}_{k}(\mathbb{R}^{j};V)$ between the space of $V$-valued $k$-homogeneous polynomials in $x_1,\ldots ,x_{j-1}$ and the space of $k$-homogeneous Dunkl monogenics in $x_1,\ldots,x_{j}$. The mapping between the two spaces is as follows.
\begin{proposition}[\cite{2015_DeBie&Genest&Vinet_ArXiv_1501.03108}]
The isomorphism between $\mathcal{P}_{k}(\mathbb{R}^{j-1})\otimes V$ and $\mathcal{M}_{k}(\mathbb{R}^{j};V)$ denoted by $\mathbf{CK}_{x_j}^{\mu_j}$ is explicitly defined by the operator
\begin{align}
\mathbf{CK}_{x_j}^{\mu_j}=\Gamma(\mu_j+1/2)\Big[\widetilde{I}_{\mu_j-1/2}(x_{j}\underline{D}_{[j-1]})+\frac{1}{2}e_j x_j \underline{D}_{[j-1]}\widetilde{I}_{\mu_j+1/2}(x_j \underline{D}_{[j-1]})\Big],
\end{align}
with $\widetilde{I}_{\alpha}(x)=(\frac{2}{x})^{\alpha}I_{\alpha}(x)$ and $I_{\alpha}(x)$ the modified Bessel function \cite{2001_Andrews&Askey&Roy}. 
\end{proposition}
\begin{remark}
When $\mu_j=0$, $\mathbf{CK}_{x_j}^{\mu_j}$ reduces to $\mathbf{CK}_{x_j}=\exp (e_j x_j \underline{D}_{[j-1]})$ which corresponds to the well-known $\mathbf{CK}$ extension operator for the standard Dirac operator, as determined in \cite{1982_Brackx&Delanghe&Sommen}.
\end{remark}
Consider $p\in \mathcal{P}_{k}(\mathbb{R}^{j-1})\otimes V$; the $\mathbf{CK}$ map can be written explicitly as
\begin{multline}
\label{CK-Explicit}
\mathbf{CK}_{x_{j}}^{\mu_j}(p)=\sum_{\ell=0}^{\lfloor \frac{k}{2}\rfloor}\frac{\Gamma(\mu_j+1/2)}{2^{2\ell}\ell !\,\Gamma(\ell+\mu_j+1/2)}\, x_{j}^{2\ell}\underline{D}_{[j-1]}^{2\ell}\,p
\\
+\frac{e_j x_j \underline{D}_{[j]}}{2}\sum_{\ell=0}^{\lfloor \frac{k-1}{2}\rfloor}\frac{\Gamma(\mu_j+1/2)}{2^{2\ell}\ell ! \Gamma(\ell+\mu_j+3/2)}x_{j}^{2\ell}\underline{D}_{[j-1]}^{2\ell}\;p.
\end{multline}
The $\mathbf{CK}$ map can be combined with the Fischer decomposition to construct an explicit basis for the space $\mathcal{M}_{k}(\mathbb{R}^{n};V)$. Let $\{v_{s}\}$ for $s\in I=\{1,\ldots, \mathrm{dim} V\}$ be a set of basis vectors for the representation space $V$ of $\mathcal{C}\ell_{n}$. One has the following tower of $\mathbf{CK}$ extensions and Fischer decompositions
\begin{align}
\label{Tower}
\begin{aligned}
\mathcal{M}_{k}(\mathbb{R}^{n};V)&\cong\mathbf{CK}_{x_n}^{\mu_n}\big[\mathcal{P}_{k}(\mathbb{R}^{n-1})\otimes V\big]\cong\mathbf{CK}_{x_n}^{\mu_n}\big[\bigoplus_{j=0}^{k}\underline{x}_{[n-1]}^{k-j}\mathcal{M}_{j}(\mathbb{R}^{n-1};V)\big]
\\
&\cong\mathbf{CK}_{x_n}^{\mu_n}\Big[\bigoplus_{j=0}^{k}\underline{x}_{[n-1]}^{k-j} \mathbf{CK}_{x_{n-1}}^{\mu_{n-1}}\big[ \mathcal{P}_{j}(\mathbb{R}^{n-2})\otimes V\big]\Big]
\\
&\cong\mathbf{CK}_{x_n}^{\mu_n}\Big[\bigoplus_{j=0}^{k}\underline{x}_{[n-1]}^{k-j} \mathbf{CK}_{x_{n-1}}^{\mu_{n-1}}\big[\bigoplus_{\ell=0}^{j}\underline{x}_{[n-2]}^{j-\ell}\mathcal{M}_{\ell}(\mathbb{R}^{n-2};V)\big]\Big]\cong\cdots
\end{aligned}
\end{align}
until one reaches $\mathcal{P}_{j_1}(\mathbb{R})\otimes V$. Since the latter space is spanned by the basis vectors $x_1^{j_1} v_{s}$ with $s\in I$ and $j_1\in \mathbb{N}$, one concludes the following.
\begin{proposition}
Let $\mathbf{j}$ be defined as $\mathbf{j}=(j_1,j_2,\ldots,j_{n-2}, j_{n-1}=k-\sum_{i=1}^{n-2}j_i)$ where $j_1,\ldots, j_{n-2}$ are non-negative integers such that $\sum_{i=1}^{n-2}j_i\leq k$. Consider the set of functions $\Psi_{\mathbf{j}}^{s}(x_1,\ldots, x_{n})$ defined by
\begin{multline}
\label{Basis-1}
\Psi_{\mathbf{j}}^{s}(x_1,\ldots, x_{n})=
\\
\mathbf{CK}_{x_{n}}^{\mu_{n}}\bigg[\underline{x}_{[n-1]}^{j_{n-1}}\mathbf{CK}_{x_{n-1}}^{\mu_{n-1}}\Big[\underline{x}_{[n-2]}^{j_{n-2}}\big[\cdots \mathbf{CK}_{x_3}^{\mu_3}[\underline{x}_{[2]}^{j_2}\mathbf{CK}_{x_2}^{\mu_2}[x_1^{j_1}]\big]\Big]\bigg] v_{s},
\end{multline}
with $s\in I$. Then the functions $\Psi_{\mathbf{j}}^{s}$ form a basis for the space $\mathcal{M}_{k}(\mathbb{R}^{n};V)$ of $k$-homogeneous Dunkl monogenics.
\end{proposition}
\begin{proof}
The result follows from the application of the tower \eqref{Tower}.
\end{proof}
\begin{remark}
As we shall always work at a given value $k$ of the degree, the dependence on $k$ of  the multi-index $\mathbf{j}$  is kept implicit.
\end{remark}
Let us now present a Lemma which characterizes the action of the operators $\underline{D}$ and $\underline{x}$ on the space of Dunkl monogenics.
\begin{lemma}
Let $M_{\ell}\in \mathcal{M}_{\ell}(\mathbb{R}^{n};V)$. Then for non-negative integers $j,k$ such that $j\leq k$ one has
\begin{align*}
\underline{D}^{2j} \underline{x}^{2k} M_{\ell} & =  
\left[\frac{2^{2j} k! \Gamma(k+ \ell + \gamma_{[n]})}{(k-j)! \Gamma(k-j+ \ell + \gamma_{[n]})} \right]\underline{x}^{2k-2j} M_{\ell}\\
\underline{D}^{2j+1} \underline{x}^{2k} M_{\ell} & = - 
\left[\frac{2^{2j+1} k! \Gamma(k+ \ell + \gamma_{[n]})}{(k-j-1)! \Gamma(k-j+ \ell + \gamma_{[n]})} \right]\underline{x}^{2k-2j-1} M_{\ell}\\
\underline{D}^{2j} \underline{x}^{2k+1} M_{\ell} & =  \left[\frac{2^{2j} k! \Gamma(k+ \ell + \gamma_{[n]}+1)}{(k-j)! \Gamma(k-j+ \ell + \gamma_{[n]}+1)}\right] \underline{x}^{2k-2j+1} M_{\ell}\\
\underline{D}^{2j+1} \underline{x}^{2k+1} M_{\ell} & = - \left[\frac{2^{2j+1} k! \Gamma(k+ \ell + \gamma_{[n]}+1)}{(k-j)! \Gamma(k-j+ \ell + \gamma_{[n]})} \right]\underline{x}^{2k-2j} M_{\ell}.
\end{align*}
\end{lemma}
\begin{proof}
The Lemma easily follows from induction and from the relations \eqref{OSP}
\end{proof}

In view of \eqref{CK-Explicit} and Lemma 13, one can expect the wavefunctions \eqref{Basis-1} to have an explicit expression in terms of hypergeometric functions. It turns out that this is indeed the case. To present the result, we introduce the following notation
\begin{align*}
\rvert \mathbf{j}_{\ell}\rvert=j_1+j_2+\cdots + j_{\ell},\qquad \rVert x_{[j]}\rVert^{2}=x_1^{2}+\cdots+x_{j}^{2}.
\end{align*}
\begin{proposition}
The basis functions $\Psi_{\mathbf{j}}^{s}$ have the expression
\begin{align}
\label{Psi-Exp}
\Psi_{\mathbf{j}}^{s}(x_1,\ldots, x_{n})=\overrightarrow{\left(\prod_{\ell=n}^{3}\right)}Q_{j_{\ell-1} }(x_{\ell}, \underline{x}_{[\ell-1]}) m_{j_1}(x_2,x_1) v_{s},
\end{align}
where
\begin{multline}
 \label{B1}
 Q_{j_{\ell-1} }(x_{\ell},\underline{x}_{[\ell-1]})=\frac{\beta! \rVert x_{[\ell]}\rVert^{2\beta}}{(\mu_{\ell}+1/2)_{\beta}}
 \\ \times
 \begin{cases}
  P_{\beta}^{(\rvert \mathbf{j}_{\ell-2}\rvert+\gamma_{[\ell-1]}-1, \mu_{\ell}-1/2)}\left(\frac{x_{\ell}^2-\rVert x_{[\ell-1]}\rVert^{2}}{\rVert x_{[\ell]} \rVert^{2}}\right)
  & \text{$j_{\ell-1}=2\beta$}
  \\[.2cm]
  \qquad  -\frac{e_{\ell} x_{\ell} \underline{x}_{[\ell-1]}}{\rVert x_{[\ell]}\rVert^{2}}\;P_{\beta-1}^{(\rvert \mathbf{j}_{\ell-2}\rvert+\gamma_{[\ell-1]},\mu_{\ell}+1/2)}\left(\frac{x_{[\ell]}^2-\rVert x_{[\ell-1]}\rVert^{2}}{\rVert x_{[\ell]}\rVert^{2}}\right)
  \\
  \\
  \underline{x}_{[\ell-1]}\;P_{\beta}^{(\rvert \mathbf{j}_{\ell-2}\rvert+\gamma_{[\ell-1]},\mu_{\ell}-1/2)}\left(\frac{x_{\ell}^2-\rVert x_{[\ell-1]}\rVert^{2}}{\rVert x_{[\ell]}\rVert^{2}}\right) 
  &
  \text{$j_{\ell-1}=2\beta+1$}
  \\
  \qquad
  -e_{\ell} x_{\ell}\left(\frac{\beta+\rvert \mathbf{j}_{\ell-2}\rvert+\gamma_{[\ell-1]}}{\beta+\mu_{\ell}+1/2}\right)\;P_{\beta}^{(\rvert \mathbf{j}_{\ell-2}\rvert+\gamma_{[\ell-1]}-1,\mu_{\ell}+1/2)}\left(\frac{x_{\ell}^2-\rVert x_{[\ell-1]}\rVert^{2}}{\rVert x_{[\ell]}\rVert^{2}}\right)
 \end{cases}
\end{multline}
and 
\begin{multline}
\label{B2}
 m_{j_1}(x_2,x_1)=\frac{(-1)^{\beta}\beta!(x_1^2+x_2^2)^{\beta}}{(\mu_2+1/2)_{\beta}}
 \\ \times
 \begin{cases}
   P_{\beta}^{(\mu_1-1/2,\mu_2-1/2)}\left(\frac{x_2^2-x_1^2}{x_1^2+x_2^2}\right)
  & \text{$j_{1}=2\beta$}
  \\[.2cm]
  \qquad 
  -\frac{e_2 e_1 x_2 x_1}{x_1^2+x_2^2}\,P_{\beta-1}^{(\mu_1+1/2, \mu_2+1/2)}\left(\frac{x_2^2-x_1^2}{x_1^2+x_2^2}\right), & 
  \\
  \\
   x_1P_{\beta}^{(\mu_1+1/2,\mu_2-1/2)}\left(\frac{x_2^2-x_1^2}{x_1^2+x_2^2}\right)
  &
  \text{$j_{1}=2\beta+1$}
  \\
  \qquad
  +\left(\frac{\beta+\mu_1+1/2}{\beta+\mu_2+1/2}\right)\,e_2e_1x_2\,P_{\beta}^{(\mu_1-1/2,\mu_2+1/2)}\left(\frac{x_2^2-x_1^2}{x_1^2+x_2^2}\right)
 \end{cases}
\end{multline}
In \eqref{B1} and \eqref{B2}, $P_{n}^{(\alpha,\beta)}(x)$ are the Jacobi polynomials  with $P_{-1}^{(\alpha,\beta)}(x)=0$ and $(a)_{n}=\frac{\Gamma(a+n)}{\Gamma(a)}$ the Pochhammer symbol \cite{2001_Andrews&Askey&Roy}.
\end{proposition}
\begin{proof}
The result follows from a long but straightforward calculation using \eqref{CK-Explicit} and Lemma 13. See also \cite{2015_DeBie&Genest&Vinet_ArXiv_1501.03108}, where a similar calculation was performed.
\end{proof}
Let us now show that the basis functions given by \eqref{Basis-1} are joint eigenfunctions of the generators of the maximal Abelian subalgebra $\mathcal{Y}_{n}$ defined in \eqref{Abelian-Sub}.
\begin{proposition}
\label{Prop-Eigen}
The wavefunctions $\Psi_{\mathbf{j}}^{s}\in \mathcal{M}_{k}(\mathbb{R}^{n};V)$ satisfy
\begin{align}
\label{C1}
\Gamma_{[\ell]}\Psi_{\mathbf{j}}^{s}(x_1,\ldots, x_{n})=\lambda_{\ell}(\mathbf{j})\Psi_{\mathbf{j}}^{s}(x_1,\ldots, x_{n}),
\end{align}
where the eigenvalues are given by
\begin{align}
\label{C2}
\lambda_{\ell}(\mathbf{j})=(-1)^{\rvert \mathbf{j}_{\ell-1}\rvert}(\rvert \mathbf{j}_{\ell-1}\rvert+\gamma_{[\ell]}-1/2).
\end{align}
\end{proposition}
\begin{proof}
For $\ell=n$, the result is equivalent to Proposition 3. For $\ell< n$, observe that $[\Gamma_{[\ell]}, \underline{D}_{[j]}]=[\Gamma_{[\ell]}, \underline{x}_{[j]}]=0$ for all $j\geq  \ell$. As a result, $\Gamma_{[\ell]}$ commutes with $\mathbf{CK}_{x_{j}}^{\mu_{j}}$ for $j-1\geq \ell$ and one can write
\begin{align*}
\Gamma_{[\ell]}\Psi_{\mathbf{j}}^{s}(x_1,\ldots, x_{n})=\mathbf{CK}_{x_{n}}^{\mu_{n}}\Big[\cdots \Gamma_{[\ell]}\mathbf{CK}_{x_{\ell}}^{\mu_{\ell}}\big[ \cdots \mathbf{CK}_{x_2}^{\mu_2}[x_1^{j_1}]\big]\Big]v_{s}.
\end{align*}
As $\mathbf{CK}_{x_{\ell}}^{\mu_{\ell}}\big[ \cdots \mathbf{CK}_{x_2}^{\mu_2}[x_1^{j_1}]\big]v_{s}\in \mathcal{M}_{\rvert \mathbf{j}_{\ell-1}\rvert }(\mathbb{R}^{\ell};V)$, the result once again follows from Proposition 3.
\end{proof}
\begin{remark}
It is easily verified from \eqref{Basis-1} that for fixed $k$ and $s$, the dimension of the space spanned by the eigenfunctions $\Psi_{\mathbf{j}}^{s}(x_1,\ldots,x_{n})$ is $\binom{n+k-2}{k}$, as expected from the isomorphism between $\mathcal{P}_{k}(\mathbb{R}^{j-1})\otimes V$  and $\mathcal{M}_{k}(\mathbb{R}^{j};V)$.
\end{remark}

\begin{remark}
If desired, the wavefunctions \eqref{Psi-Exp} can be normalized to satisfy an orthogonality relation of the form
\begin{align*}
\smashoperator{\int_{\mathbb{S}^{n-1}}} |x_1|^{2\mu_1}\mathrm{d}x_1\cdots  |x_n|^{2\mu_n}\mathrm{d}x_{n}\;\big[\Psi_{\mathbf{\mathbf{j}'}}^{s'}(x_1,\ldots, x_{n})\big]^{*}\cdot \Psi_{\mathbf{j}}^{s}(x_1,\ldots, x_{n})=\delta_{\mathbf{j}\mathbf{j}'}\delta_{ss'},
\end{align*} 
where $*$ denotes complex conjugation and where $\cdot$ stands for an appropriately defined inner product on the representation space $V$. One should also take $e_i^{*}=-e_i$ for the Clifford generators and set $(ab)^{*}=b^{*}a^{*}$ for $a,b\in \mathcal{C}\ell_{n}$. The normalization coefficients were given in \cite{2015_DeBie&Genest&Vinet_ArXiv_1501.03108} for the $n=3$ case and are readily generalized for $n >3$. The orthogonality property can be proven directly by induction on $n$ and follows from the well-known orthogonality relation satisfied by the Jacobi polynomials \cite{2001_Andrews&Askey&Roy}.
\end{remark}
From the common eigenfunctions $\Psi_{\mathbf{j}}^{s}(x_1,\ldots,x_{n})$ of the Abelian subalgebra $\mathcal{Y}_{n}$, one can easily construct a set of joint eigenfunctions for the Abelian subalgebra $\pi \mathcal{Y}_{n}$, where $\pi$ is a permutation of $S_{n}$. Indeed, it is directly verified that $\pi \Psi_{\mathbf{j}}^{s}(x_1,\ldots,x_{n})$ are common eigenfunctions for $\pi \mathcal{Y}_{n}$, where the permutation $\pi$ acts on $\Psi_{\mathbf{j}}^{s}(x_1,\ldots,x_{n})$ by simultaneously permuting the variables $x_1,\ldots, x_{n}$, the Clifford generators $e_1,\ldots, e_{n}$ and the parameters $\mu_1,\ldots, \mu_{n}$. In light of this observation, consider the functions defined as
\begin{align*}
\Phi_{\mathbf{j'}}^{s}(x_1,\ldots,x_{n})=\pi \Psi_{\mathbf{j}'}^{s}(x_1,\ldots,x_{n})\qquad \pi=(123\cdots n),
\end{align*}
which are common eigenfunctions for $\mathcal{Z}_{n}$, as defined in \eqref{Abelian-Sub-2}. It is manifest that the functions $\Phi_{\mathbf{j'}}^{s}(x_1,\ldots,x_{n})$ also form a basis for $\mathcal{M}_{k}(\mathbb{R}^{n};V)$. As both $\Phi_{\mathbf{j'}}^{s}(x_1,\ldots,x_{n})$ and $\Psi_{\mathbf{j}}^{s}(x_1,\ldots,x_{n})$ can be normalized to form finite-dimensional orthonormal bases for $\mathcal{M}_{k}(\mathbb{R}^{n};V)$, there exists a unitary transformation, or intertwining operator, connecting these two bases. In \cite{2015_DeBie&Genest&Vinet_ArXiv_1501.03108}, the matrix elements of this intertwining operator were seen to involve the Bannai--Ito polynomials. In the $n$-dimensional case, one can expect that the unitary transformation between $\Phi_{\mathbf{j'}}^{s}(x_1,\ldots,x_{n})$ and $\Psi_{\mathbf{j}}^{s}(x_1,\ldots,x_{n})$ involves multivariate polynomials, which could be viewed as multivariate analogs of the Bannai--Ito polynomials. Since a theory of Bannai--Ito polynomials in many variables is yet to be worked out, we formulate the following as a conjecture.
\begin{conjecture}
The matrix elements of the intertwining operator between the basis functions $\Psi_{\mathbf{j}}^{s}(x_1,\ldots,x_{n})$ and $\Phi_{\mathbf{j'}}^{s}(x_1,\ldots,x_{n})$ are expressed in terms of orthogonal polynomials which could be taken to define the multivariate extension of the Bannai--Ito polynomials.
\end{conjecture}
\section{Action of the symmetry algebra}
In this section, the action of the symmetry algebra $\mathcal{A}_{n}$ on the basis functions $\Psi_{\mathbf{j}}^{s}(x_1,\ldots,x_{n})$ is studied. A set of raising and lowering operators is introduced and their actions are computed explicitly. Using these results, it is observed that $\mathcal{A}_{n}$ acts irreducibly on the space $\mathcal{M}_{k}(\mathbb{R}^{d-1};V)$ of $k$-homogeneous Dunkl monogenics.

For $\ell\in [n-2]$, consider the set of operators $K_{\ell}^{\pm}$ defined as
\begin{multline}
\label{Kpm}
K_{\ell}^{\pm}=
\\
(\Gamma_{\{\ell+1,\ell+2\}}\pm \Gamma_{[\ell+2]\setminus\{\ell+1\}})(\Gamma_{[\ell+1]}\mp 1/2)-(\Gamma_{\{\ell+2\}}\pm \Gamma_{[\ell+2]})(\Gamma_{[\ell]}\pm \Gamma_{\{\ell+1\}}).
\end{multline}
For simplicity, we write $K_{\ell}^{\pm}=A_{\ell}^{\pm}+B_{\ell}^{\pm}$ with
\begin{align*}
A_{\ell}^{\pm}&=(\Gamma_{\{\ell+1,\ell+2\}}\pm \Gamma_{[\ell+2]\setminus\{\ell+1\}})(\Gamma_{[\ell+1]}\mp 1/2),
\\
B_{\ell}^{\pm}&=-(\Gamma_{\{\ell+2\}}\pm \Gamma_{[\ell+2]})(\Gamma_{[\ell]}\pm \Gamma_{\{\ell+1\}}).
\end{align*}
The operators  $K_{\ell}^{\pm}$ have nice covariance properties with respect to the Abelian subalgebra $\mathcal{Y}_{n}$. These properties are described in the following proposition.
\begin{proposition}
\label{prop9}
Let $j\in [n]$ and $\ell\in [n-2]$. One has
\begin{align*}
[K_{\ell}^{\pm}, \Gamma_{[j]}]=0,\qquad j\neq \ell+1.
\end{align*}
Conversely when $j=\ell+1$ one has
\begin{align*}
\{K_{\ell}^{\pm}, \Gamma_{[\ell+1]}\}=\pm K_{\ell}^{\pm}.
\end{align*}
\end{proposition}
\begin{proof}
The first relation follows immediately from the definition of $K_{\ell}^{\pm}$ and from Lemma \ref{Lemma-Comm}. For the second identity, one has
\begin{align*}
\{K_{\ell}^{\pm},\Gamma_{[\ell+1]}\}&=2B_{\ell}^{\pm}\Gamma_{[\ell+1]}+\{\Gamma_{\{\ell+1,\ell+2\}}\pm \Gamma_{[\ell+2]\setminus\{\ell+1\}}, \Gamma_{[\ell+1]}\}(\Gamma_{[\ell+1]}\mp 1/2)
\\
&=2B_{\ell}^{\pm}\Gamma_{[\ell+1]}\pm A_{\ell}^{\pm}+2(\Gamma_{[\ell]}\pm \Gamma_{\{\ell+1\}})(\Gamma_{\{\ell+2\}}\pm\Gamma_{[\ell+2]})(\Gamma_{[\ell+1]}\mp 1/2)
\\
&=2B_{\ell}^{\pm}\Gamma_{[\ell+1]}\pm A_{\ell}^{\pm}-2B_{\ell}^{\pm}(\Gamma_{[\ell+1]}\mp 1/2)=\pm K_{\ell}^{\pm},
\end{align*}
where the relations \eqref{BI-Relations} were used.
\end{proof}
In order to characterize the action of the operators $K_{\ell}^{\pm}$ on the basis functions $\Psi_{\mathbf{j}}^{s}(x_1,\ldots,x_{n})$, we shall need the following technical Lemma.
\begin{lemma}
One has the identity
\begin{align}
\label{I1}
\Gamma_{\{\ell+1,\ell+2\}}^{2}+\Gamma_{[\ell+2]\setminus \{\ell+1\}}^2=\Gamma_{[\ell+2]}^2-\Gamma_{[\ell+1]}^2+\Gamma_{[\ell]}^2+\Gamma_{\{\ell+2\}}^2+\Gamma_{\{\ell+1\}}^2-1/4.
\end{align}
\end{lemma}
\begin{proof}
The proof follows from the application of Lemma \ref{Lem-Relation}. One first writes
\begin{multline*}
\Gamma_{\{\ell+1,\ell+2\}}^{2}+\Gamma_{[\ell+2]\setminus \{\ell+1\}}^2=
\\
(\Gamma_{\{\ell+1,\ell+2\}}^2+\sum_{i\in [\ell]}\Gamma_{\{i,\ell+1\}}^2+\Gamma_{[\ell+2]\setminus \{\ell+1\}}^2)-(\sum_{i\in [\ell]}\Gamma_{\{i,\ell+1\}}^2+\Gamma_{[\ell]}^2)+\Gamma_{[\ell]}^2.
\end{multline*}
To get the result, one first uses \eqref{Relation} on the first two terms. Applying \eqref{Relation} again on the result and simplifying the expression yields \eqref{I1}.
\end{proof}
Quite strikingly, the squares of the operators $K_{\ell}^{\pm}$ have an exact factorization.
\begin{proposition}
\label{Prop-Square}
One has
\begin{multline*}
[K_{\ell}^{\pm}]^2=
(\Gamma_{[\ell+2]} -\Gamma_{[\ell+1]} \pm \Gamma_{\{\ell+2\}} \pm 1/2)
(\Gamma_{[\ell+2]} + \Gamma_{[\ell+1]} \pm \Gamma_{\{\ell+2\}} \mp 1/2)
\\ \times
(\Gamma_{[\ell]} - \Gamma_{[\ell+1]} \pm \Gamma_{\{\ell+1\}} \pm 1/2)
(\Gamma_{[\ell]} + \Gamma_{[\ell+1]} \pm \Gamma_{\{\ell+1\}} \mp 1/2).
\end{multline*}
\end{proposition}
\begin{proof}
Since $A_{\ell}^{\pm}$ and $B_{\ell}^{\pm}$ commute with one another, one has
\begin{align*}
(K_{\ell}^{\pm})^2=(A_{\ell}^{\pm})^2+(B_{\ell}^{\pm})^2+2A_{\ell}^{\pm} B_{\ell}^{\pm}.
\end{align*}
Upon observing that one has
\begin{align*}
(\Gamma_{[\ell+1]}\mp 1/2)(\Gamma_{\{\ell+1,\ell+2\}}\pm \Gamma_{[\ell+2]\setminus \{\ell+1\}})=-A_{\ell}^{\pm}-2B_{\ell}^{\pm},
\end{align*}
which follows directly from the commutation relations \eqref{BI-Relations}, one can write
\begin{align}
\label{TMP1}
[K_{\ell}^{\pm}]^2=(B_{\ell}^{\pm})^{2}-(\Gamma_{\{\ell+1,\ell+2\}}\pm \Gamma_{[\ell+2]\setminus \{\ell+1\}})^2 (\Gamma_{[\ell+1]}\mp 1/2)^2.
\end{align}
Using \eqref{I1}, one has
\begin{multline*}
(\Gamma_{\{\ell+1,\ell+2\}}\pm \Gamma_{[\ell+2]\setminus \{\ell+1\}})^2=(\Gamma_{\{\ell+2\}}\pm \Gamma_{[\ell+2]})^2
\\
+(\Gamma_{[\ell]}\pm \Gamma_{\{\ell+1\}})^2-(\Gamma_{[\ell+1]}\mp 1/2)^2.
\end{multline*}
Upon inserting the above expression in \eqref{TMP1}, one gets
\begin{multline*}
[K_{\ell}^{\pm}]^2=\left[(\Gamma_{[\ell+2]}\pm\Gamma_{\{\ell+2\}})^2-(\Gamma_{[\ell+1]}\mp 1/2)^2\right]
\\
\times \left[(\Gamma_{[\ell]}\pm \Gamma_{\{\ell+1\}})^2-(\Gamma_{[\ell+1]}\mp 1/2)^2\right].
\end{multline*}
Factoring each term gives the anticipated result.
\end{proof}
\begin{lemma}
The action of $(K_{\ell}^{\pm})^2$ on the $\Psi_{\mathbf{j}}^{s}(x_1,\ldots,x_{n})$ is given by
\begin{align*}
\left( K_\ell^{\pm} \right)^2 \Psi_{\mathbf{j}}^{s}(x_1,\ldots,x_{n}) =  \alpha_{\pm}^{\ell}(\mathbf{j}) \Psi_{\mathbf{j}}^{s}(x_1,\ldots,x_{n}),
\end{align*}
where
\begin{align}
\label{Coef}
\begin{aligned}
\alpha_{\pm}^{\ell}(\mathbf{j})&= 
(\lambda_{\ell+2}(\mathbf{j})  - \lambda_{\ell+1}(\mathbf{j}) \pm \mu_{\ell+2}  \pm 1/2)
(\lambda_{\ell+2}(\mathbf{j})  + \lambda_{\ell+1}(\mathbf{j}) \pm \mu_{\ell+2} \mp 1/2 )\\
& \quad \times
(\lambda_{\ell}(\mathbf{j}) - \lambda_{\ell+1}(\mathbf{j})\pm \mu_{\ell+1}  \pm 1/2)
(\lambda_{\ell}(\mathbf{j}) + \lambda_{\ell+1}(\mathbf{j}) \pm \mu_{\ell+1} \mp 1/2).
\end{aligned}
\end{align}
The coefficients $\alpha_{\pm}^{\ell}(\mathbf{j})$ are non-negative for all admissible $\mathbf{j}$. In particular, they are non-zero except in the following cases
\begin{align}
\label{vanish}
\begin{cases}
j_{\ell+1}=0\;\text{and}\;|\mathbf{j}_{\ell}|\;\text{even} & \Rightarrow\alpha_{-}^{\ell}(\mathbf{j})=0
\\
j_{\ell+1}=0\;\text{and}\;|\mathbf{j}_{\ell}|\;\text{odd} & \Rightarrow\alpha_{+}^{\ell}(\mathbf{j})=0
\\
j_{\ell}=0\;\text{and}\;|\mathbf{j}_{\ell+1}|\;\text{even} & \Rightarrow\alpha_{+}^{\ell}(\mathbf{j})=0
\\
j_{\ell}=0\;\text{and}\;|\mathbf{j}_{\ell+1}|\;\text{odd} & \Rightarrow\alpha_{-}^{\ell}(\mathbf{j})=0
\end{cases}
\end{align}
\end{lemma}
\begin{proof}
The eigenvalue follows from the combination of Propositions \ref{Prop-Eigen} and \ref{Prop-Square}. The vanishing of the coefficients $\alpha_{\pm}^{\ell}(\mathbf{j})$ is determined directly with the help of the formula \eqref{C2} for the eigenvalues. The fact that they are positive everywhere else follows from the fact that the indices $j_{\ell}$ are non-negative integers and from the assumption that $\mu_i>0$.
\end{proof}

We now derive the action of the operators $K_{\ell}^{\pm}$ on the basis functions. To describe the result, we introduce the following notation. Let $\mathbf{h}_i=(0,\ldots, 1, -1,\ldots,0)$ be the vector with $n-1$ entries that has a $1$ at position $i$, a $-1$ at position $i+1$, and zeros elsewhere. We also define $\bar{\mathbf{h}}_{n-2}=(0,\ldots,1,0)$.

\begin{proposition}  The action of the operators $K_{\ell}^{\pm}$ on the functions $\Psi_{\mathbf{j}}^{s}(x_1,\ldots,x_{n})$ of $\mathcal{M}_{k}(\mathbb{R}^{n};V)$ is as follows. For $\ell\in \{1,2,\ldots,n-3\}$, one has
\begin{align}
\label{rr}
K_{\ell}^{\pm}\Psi_{\mathbf{j}}^{s}(x_1,\ldots,x_{n})\propto
\begin{cases}
\Psi_{\mathbf{j}\pm h_{\ell}}^{s}(x_1,\ldots, x_{n}) & \text{$|\mathbf{j}_{l}|$ is odd}
\\
\Psi_{\mathbf{j}\mp h_{\ell}}^{s}(x_1,\ldots,x_{n}) & \text{$|\mathbf{j}_{l}|$ is even}
\end{cases}
\end{align}
while for $\ell=n-2$, one has
\begin{align}
\label{rr2}
K_{n-2}^{\pm}\Psi_{\mathbf{j}}^{s}(x_1,\ldots,x_{n})\propto
\begin{cases}
\Psi_{\mathbf{j}\pm \bar{h}_{n-2}}^{s}(x_1,\ldots, x_{n}) & \text{$|\mathbf{j}_{n-2}|$ is odd}
\\
\Psi_{\mathbf{j}\mp \bar{h}_{n-2}}^{s}(x_1,\ldots,x_{n}) & \text{$|\mathbf{j}_{n-2}|$ is even}
\end{cases}.
\end{align}
\end{proposition}
\begin{proof}
Consider the case where $\ell\in [n-3]$. Since $K_{\ell}^{\pm}$ commutes with $\Gamma_{[i]}$ for $i=1,\ldots,\ell$, it preserves the quantum numbers $j_1, j_2,\ldots, j_{\ell-1}$. Since $K_{\ell}^{\pm}$ also commutes with $\Gamma_{k}$ for $k=\ell+2,\ldots, n$, it must preserve the sum $j_{\ell}+j_{\ell+1}$, the individual quantum numbers $j_{\ell+2},\ldots, j_{n-2}$ and the total degree $k$ which implicitly appears in the quantum number $j_{n-1}$. Thus, $K_{\ell}^{\pm}$ can only act on $\Psi_{\mathbf{j}}^{s}(x_1,\ldots,x_{n})$ by changing $j_{\ell}$ and $j_{\ell+1}$ while preserving the value of their sum. Furthermore, it follows from Propositions \ref{Prop-Eigen} and \ref{prop9}  that 
\begin{align*}
\Gamma_{[\ell+1]} K_{\ell}^{\pm}\Psi_{\mathbf{j}}^{s}(x_1,\ldots, x_{n})&=\Big[\{\Gamma_{[\ell+1]}, K_{\ell}^{\pm}\}-K_{\ell}^{\pm}\Gamma_{[\ell+1]}\Big]\Psi_{\mathbf{j}}^{s}(x_1,\ldots, x_{n})
\\
&=\big[\pm 1-\lambda_{\ell+1}(\mathbf{j})\big] K_{\ell}^{\pm}\Psi_{\mathbf{j}}^{s}(x_1,\ldots, x_{n}).
\end{align*}
The above equation means that $K_{\ell}^{\pm}\Psi_{\mathbf{j}}^{s}(x_1,\ldots, x_{n})$ is an eigenfunction of the operator $\Gamma_{[\ell+1]}$ with eigenvalue $[\pm 1-\lambda_{\ell+1}(\mathbf{j})]$.  Upon examining the explicit form \eqref{C2} of the eigenvalues $\lambda_{\ell+1}(\mathbf{j})$ while taking into account that only $j_{\ell}$ and $j_{\ell+1}$ are allowed to change, one writes
\begin{align*}
\big[1-\lambda_{\ell+1}(\mathbf{j})\big]&=
\begin{cases}
\lambda_{\ell+1}(\mathbf{j}+\bar{h}_{\ell}) & \text{$|\mathbf{j}_{\ell}|$ odd}
\\
\lambda_{\ell+1}(\mathbf{j}-\bar{h}_{\ell}) & \text{$|\mathbf{j}_{\ell}|$ even}
\end{cases}
\\
-\big[1+\lambda_{\ell+1}(\mathbf{j})\big]&=
\begin{cases}
\lambda_{\ell+1}(\mathbf{j}-\bar{h}_{\ell}) & \text{$|\mathbf{j}_{\ell}|$ odd}
\\
\lambda_{\ell+1}(\mathbf{j}+\bar{h}_{\ell}) & \text{$|\mathbf{j}_{\ell}|$ even}
\end{cases}
\end{align*}
which yields \eqref{rr}. The proof of \eqref{rr2} proceeds similarly, with the difference of notation coming from the special role played by the component $j_{n-1}$.
\end{proof}
\begin{remark}
The exact proportionality factor in \eqref{rr} and \eqref{rr2} is not needed for our purposes. This factor can be evaluated explicitly in terms of the coefficients $\alpha_{\pm}^{\ell}(\mathbf{j})$ and is non-zero in view of Lemma 15. This could be done, for example, by demanding that the operators $K_{\ell}^{\pm}$ be skew-hermitian in the orthonormalized basis arising from the wavefunctions $\Psi_{\mathbf{j}}^{s}(x_1,\ldots, x_{n})$; this is a natural choice if one wishes to interpret the symmetry operators $\Gamma_{A}$ as quantum mechanical observables.
\end{remark}
\begin{remark} 
As can be seen from \eqref{C1}, \eqref{rr}, and \eqref{rr2}, the quantum number $s$ is not associated to any of the symmetry operators $\Gamma_{A}$. Since $s$ is tied to the choice of a basis for the representation space $V$, the operator associated to this quantum number depends on that choice. For example, one could take the vectors $v_{s}$ as eigenvectors of $e_1e_2$ in the representation $V$. The corresponding symmetry operator would then be $P=e_1e_2 r_1 r_2$. Note that the presence of $r_1r_2$ is necessary to ensure that $P$ commutes with every generator of $\mathcal{Y}_{n}$.
\end{remark}
From the actions \eqref{rr} and \eqref{rr2} of the raising and lowering operators, the expression \eqref{Kpm} and the eigenvalue equations of Proposition \ref{Prop-Eigen}, it is apparent that one can reconstruct the action of any generator $\Gamma_{A}$ on the basis functions $\Psi_{\mathbf{j}}^{s}(x_1,\ldots,x_{n})$. However, since the relation between raising/lowering operators is quadratic, there is no systematic way of doing so and the end result is often very involved. We now present the last proposition for this section.
\begin{proposition}
For a fixed $s$, the higher rank Bannai--Ito algebra $\mathcal{A}_{n}$ acts irreducibly on the space
\begin{align*}
\mathcal{M}_{k}^{s}(\mathbb{R}^{n};V)=\mathrm{Span}\{\Psi_{\mathbf{j}}^{s}(x_1,\ldots,x_{n})\},
\end{align*}
with $\mathbf{j}=(j_1,\ldots,j_{n-1})$,  where $j_1,\ldots, j_{n-1}$ are non-negative integers with $j_{n-1}$ defined as $j_{n-1}=k-\sum_{i=1}^{n-2}j_i$ and $\sum_{i=1}^{n-2}j_{i}\leq k$. The space of $k$-homogeneous Dunkl monogenics in $\mathbb{R}^{n}$ splits in $\dim V$ copies of this irreducible subspace
\begin{align*}
\mathcal{M}_{k}(\mathbb{R}^{n};V)=\bigoplus_{s=1}^{\dim V}\mathcal{M}_{k}^{s}(\mathbb{R}^{n};V)
\end{align*}
\end{proposition}
\begin{proof}
The decomposition of $\mathcal{M}_{k}(\mathbb{R}^{n};V)$ as a direct sum of $\mathcal{M}_{k}^{s}(\mathbb{R}^{n};V)$ is immediate from remark 10. The irreducibility of the representation is manifest from the form of the actions \eqref{rr} and \eqref{rr2}.
\end{proof}
\section{A scalar model of the higher rank Bannai--Ito algebra}
In this section, an alternative realization of the higher rank Bannai--Ito algebra is proposed. The model does not involve Clifford generators, and is hence ``scalar''.

Let $D$, $X$ be defined as
\begin{align*}
D=\sum_{i=1}^{n}T_i R_i,\qquad X=\sum_{i=1}^{n} X_i R_i,
\end{align*}
where $R_i=\prod_{j=i+1}^{n}r_{j}$. It is readily verified that $D$ also provides a factorization of the Dunkl--Laplace operator; one has indeed
\begin{align*}
D^2=\Delta,\qquad X^2=\sum_{i=1}^{n}x_i^2.
\end{align*}
Similarly, one can define intermediate operators in the following way. For $A\subset [n]$, we consider
\begin{align*}
D_{A}=\sum_{i\in A}D_i R_i,\qquad X_{A}=\sum_{i\in A}x_i R_i.
\end{align*}
One can verify that the operators $D_{A}$ and $X_{A}$ together with the Euler operator $\mathbb{E}_{A}$, the Dunkl-Laplace operator $\Delta_{A}$ and the square-radius operator $\rVert x_{A}\rVert^{2}$ provide a realization of the $\mathfrak{osp}(1|2)$ Lie superalgebra as in \eqref{OSP}. Following the steps of the first section, one can introduce the sCasimir operator
\begin{align*}
S_{A}=\frac{1}{2}\Big([D_{A},X_{A}]-1\Big).
\end{align*}
which satisfies $\{S_{A}, D_{A}\}=0$ and $\{S_{A},X_{A}\}=0$. Upon defining,
\begin{align*}
\Gamma_{A}=S_{A}\prod_{k\in A}r_{k},
\end{align*}
it follows that $[\Gamma_{A},D_{A}]=[\Gamma_{A}, X_{A}]=0$. With these definitions, it can be verified that the operators $\Gamma_{A}$ satisfy the commutation relations \eqref{BI-Relations}. Moreover, one can confirm that the realization-dependent values taken by the Casimir operators \eqref{Relation} and \eqref{Cas2-Value} are also identical. Furthermore, it is possible to adapt the $\mathbf{CK}$ extension operator to the scalar case.
\begin{remark}
The $n=3$ case of this scalar model was already examined in \cite{2015_Genest&Vinet&Zhedanov_CommMathPhys_336_243}.
\end{remark}
\begin{remark}
The Proposition 12 and the existence of the scalar realization indicate that the Clifford algebra representation space $V$ does not play a significant role.
\end{remark}
\begin{remark}
One might ask why the focus was put on the Dirac--Dunkl model rather than the scalar model, which appears simpler. It turns out that the Dirac--Dunkl is advantageous in the construction of the symmetry algebra. Moreover, this model can be viewed as a deformation of the classical Dirac equation by Dunkl operators, and is hence more interesting ``physically''. However, let us point out that the scalar realization is likely to be more adapted for the study of Conjecture 1.
\end{remark}
\section{Conclusion}
Summing up, we studied the kernel of the $n$-dimensional Dirac--Dunkl operator associated to the reflection group $\mathbb{Z}_2^{n}$ as well as the corresponding Dirac--Dunkl equation on the $(n-1)$-sphere. We obtained the symmetries of this equation and shown that they generate a higher rank extension of the Bannai--Ito algebra. Moreover, a basis for the polynomial null-solutions of the Dirac--Dunkl operator was constructed explicitly in terms of Jacobi polynomials. Furthermore, the action of the higher rank Bannai-Ito algebra on this basis was given through raising and lowering operators for this algebra.

The results of this paper call for a further study of the structure and the properties of the higher rank Bannai--Ito algebra $\mathcal{A}_{n}$. One of the most interesting question in that regard is that of the classification of its unitary irreducible representations. This has not been investigated so far, even in the rank-one case, which corresponds to $\mathcal{A}_3$. In light of the inclusion $\mathcal{A}_{n}\subset \mathcal{A}_{n-1}$, the representations of $\mathcal{A}_{n}$ could be studied using induction. Another interesting avenue is the analysis of the symmetries of Dirac--Dunkl operators associated to other reflection groups.
\section*{Acknowledgements}
\noindent
HDB is grateful for the hospitality extended to him by the Centre de recherches math\'ematiques, where part of this research was carried. The research of HDB is supported by the Fund for Scientific Research-Flanders (FWO-V), project ``Construction of algebra realisations using Dirac-operators'', grant G.0116.13N. VXG holds a postdoctoral fellowship from the Natural Science and Engineering Research Council of Canada (NSERC). The research of LV is supported in part by NSERC. 

\end{document}